\documentclass[conference]{IEEEtran}

\usepackage{amsmath, amsthm, amssymb}
\usepackage{latexsym}
\usepackage{graphicx}
\usepackage{verbatim}
\usepackage{mathrsfs}
\usepackage{float}
\usepackage[lofdepth, lotdepth]{subfig}
\usepackage{array}
\newcolumntype{L}[1]{>{\raggedright\let\newline\\\arraybackslash\hspace{0pt}}m{#1}}
\newcolumntype{C}[1]{>{\centering\let\newline\\\arraybackslash\hspace{0pt}}m{#1}}
\newcolumntype{R}[1]{>{\raggedleft\let\newline\\\arraybackslash\hspace{0pt}}m{#1}}
\usepackage{url}
\usepackage{algorithm}
\usepackage[noend]{algorithmic}
\usepackage{cite}
\usepackage[table]{xcolor}
\usepackage{enumerate}
\usepackage{eucal}
\usepackage{stmaryrd}
\usepackage{mathtools}
\usepackage{pgf}

\usepackage{tikz}
\usetikzlibrary{arrows, arrows.meta,calc,positioning,shapes.arrows}
\usetikzlibrary{decorations.pathreplacing}
\usetikzlibrary{chains,fit,shapes}
\usetikzlibrary{mindmap,trees}

\usepackage[margin=0.50in]{geometry}
\usepackage{mathtools}
\DeclarePairedDelimiter\ceil{\lceil}{\rceil}
\DeclarePairedDelimiter\floor{\lfloor}{\rfloor}

\theoremstyle{plain}
\newtheorem{theorem}{Theorem}
\newtheorem{proposition}[theorem]{Proposition}
\newtheorem{lemma}[theorem]{Lemma}

\newtheorem{corollary}[theorem]{Corollary}

\theoremstyle{definition}
\newtheorem{definition}[theorem]{Definition}
\newtheorem{example}[theorem]{Example}
\newtheorem{remark}[theorem]{Remark}

\newcommand{\C}{{\mathcal C}}
\newcommand{\D}{{\mathcal D}}

\newcommand{\G}{{\mathcal G}}

\newcommand{\Q}{{\mathcal Q}}
\newcommand{\X}{{\mathcal X}}

\newcommand{\Y}{{\mathcal Y}}

\newcommand{\cS}{{\mathcal S}}
\newcommand{\T}{{\mathcal T}}

\DeclareMathAlphabet{\mathbfsl}{OT1}{ppl}{b}{it} %{OT1}{cmr}{bx}{it}

\newcommand{\ba}{{\mathbfsl a}}

\newcommand{\bara}{\overline{\mathbfsl a}}

\newcommand{\bu}{{\mathbfsl u}}
\newcommand{\bv}{{\mathbfsl v}}

\newcommand{\by}{{\mathbfsl y}}

\newcommand{\bd}{{\mathbfsl{d}}}

\newcommand{\bw}{{\mathbfsl{w}}}

\newcommand{\br}{{\mathbfsl{r}}}
\newcommand{\bx}{{\mathbfsl{x}}}
\newcommand{\bz}{{\mathbfsl{z}}}

\newcommand{\be}{{\mathbfsl e}}

\newcommand{\bbZ}{{\mathbb Z}}

\renewcommand{\ge}{\geqslant}
\renewcommand{\le}{\leqslant}

\newcommand{\syn}{{\rm Syn}}

\newcommand{\CSVT}{{\cal C}_{\rm CSVT}}

\newcommand{\etal}{{\em et al.}}

\begin{document}

\pagestyle{empty}

\title{Optimal Reconstruction Codes for Deletion Channels\\[-5mm]}
\author{\IEEEauthorblockN{
Johan Chrisnata\IEEEauthorrefmark{1,2},
Han Mao Kiah\IEEEauthorrefmark{2}, 
and Eitan Yaakobi\IEEEauthorrefmark{1}}\\[-5mm]
%and Kui Cai\IEEEauthorrefmark{2}}\\[-3mm]
\IEEEauthorblockA{
\IEEEauthorrefmark{1}%
Department of Computer Science, Technion --- Israel Institute of Technology, Haifa, 3200003 Israel\\
\IEEEauthorrefmark{2}%
School of Physical and Mathematical Sciences, Nanyang Technological University, Singapore 637371\\
Emails: johanchr001@ntu.edu.sg, hmkiah@ntu.edu.sg, yaakobi@cs.technion.ac.il\\[-5mm]
}
}

\maketitle

%----------------------------------------------------------------------------------------------------------------------------------------------------------------------------
%SECTION ABSTRACT 

\hspace{-3.5mm}\begin{abstract}
The sequence reconstruction problem, introduced by Levenshtein in 2001, considers a communication scenario where the sender transmits a codeword from some codebook and the receiver obtains multiple noisy reads of the codeword.  
Motivated by modern storage devices, we introduced a variant of the
problem where the number of noisy reads $N$ is fixed (Kiah \etal{ }2020). 
Of significance, for the single-deletion channel, 
using $\log_2\log_2 n +O(1)$ redundant bits, we designed a reconstruction code of length $n$ that reconstructs codewords from two distinct noisy reads.

In this work, we show that $\log_2\log_2 n -O(1)$ redundant bits are necessary for such reconstruction codes, 
thereby, demonstrating the optimality of our previous construction.
Furthermore, we show that these reconstruction codes can be used in $t$-deletion channels (with $t\ge 2$) to uniquely reconstruct codewords from $n^{t-1}+O\left(n^{t-2}\right)$ distinct noisy reads.
\end{abstract}

\vspace{-3mm}

% INTRODUCTION

\section{Introduction}

As our data needs surge, new technologies emerge to store these huge datasets. 
Interestingly, besides promising ultra-high storage density, certain emerging storage media, 
such as DNA based storage \cite{Church.etal:2012,Goldman.etal:2013,Yazdi.etal:2015b, Organick2017} and racetrack memories \cite{Parkin.2008, Zhang.2016, chee2018coding}, rely on technologies that provide users with multiple cheap, albeit noisy, reads.
In our companion paper \cite{Kiah.arxiv.2020}, we proposed a {\em coding solution} to leverage on these multiple reads to increase the information capacity,
% of these next-generation devices, or
or equivalently, reduce the number of redundant bits.

Our code design problem is based on the {\em sequence reconstruction problem}, formulated by Levenshtein \cite{Levenshtein.2001}. 
In Levenshtein's seminal work, he considers a communication scenario where the sender {transmits} a codeword from some codebook  and the receiver obtains multiple noisy reads of the codeword.  
The common setup assumes the codebook to be the entire space and the problem is to determine the minimum number of distinct reads $N$ that is required to reconstruct the transmitted codeword.
In constrast, in our problem, the parameter $N$ is fixed and our task is to design a {\em codebook} such that every codeword can be uniquely reconstructed from any $N$ distinct noisy reads.

Hence, our fundamental problem is then: how large can this codebook be? 
Or equivalently, what is the {\em minimum redundancy}?
Modifying a code construction in \cite{chee2018coding}, we provided in~\cite{Kiah.arxiv.2020} a number of reconstruction codes for the single-edit channel and its variants with $\log_2\log_2 n+O(1)$ bits of redundancy.
In this paper, we focus on the converse of the problem and demonstrate that $\log_2\log_2 n - O(1)$ redundant bits are {\em necessary}.
To ease our exposition, we focus on channels with {\em deletions only} and our first contribution is to demonstrate this lower bound on redundancy for the case $N=2$ and single deletions.

In our proof, we characterize the conditions when two single-deletion balls have intersection size two
(i.e. when two different codewords result in two noisy reads through the single-deletion channel).
In this same spirit, we determine when two single-deletion balls have intersection size one.
Using this characterization, we show that the same reconstruction code for the single-deletion channel 
can be used to uniquely reconstruct codewords with 
approximately half the number of reads (as compared to the case with no coding) for the $t$-deletions channel with $t\ge 2$. 
%\e{a significantly fewer number of} reads for the $t$-deletions channel with $t\ge 2$. %\footnote{\e{can we be more explicit with this sentence?}}.
We formally describe our problem and results in the next section.

%\begin{itemize}
%\item Motivate problem: applications - DNA based storage and racetrack memory
%\item Focus on channel with deletions
%\item Companion paper demonstrates the existence of $(n,2;D_1)$-reconstruction codes $\C$ with $\log\log n+O(1)$ redundant bits.
%\item Here, we demonstrate that these codes are optimal in terms of redundancy. 
%\item Also, we show that the same code can be used to reconstruct with significantly lesser reads for general $t$.
%\end{itemize}

\section{Preliminaries}

Consider a data storage scenario described by an error-ball function.
Formally, given an input space $\X$ and an output space $\Y$, 
an {\em error-ball} function $B$ maps a {\em word} $\bx\in \X$ 
to a subset of {\em noisy reads} $B(\bx)\subseteq \Y$.
Given a code $\C\subseteq \X$, we define the {\em read coverage} of $\C$, 
denoted by $\nu(\C;B)$, to be the quantity
\begin{equation*}\label{eq:read}
\nu(\C;B)\triangleq \max \Big\{|B(\bx)\cap B(\by)| : \, \bx,\by\in \C, \, \bx\ne\by \Big\}\, .
\end{equation*}
\noindent 
In other words, $\nu(\C;B)$ is the maximum intersection between the error-balls of any two codewords in $\C$. 
The quantity $\nu(\C;B)$ was introduced by Levenshtein \cite{Levenshtein.2001}, 
where he showed that the number of reads%
\footnote{In the original paper, Levenshtein used the term ``channels'', instead of reads. 
Here, we used the term ``reads'' to reflect the data storage scenario.}
required to reconstruct a codeword from $\C$ is at least $\nu(\C;B)+1$.
The problem to determine $\nu(\C;B)$ is referred to as the {\em sequence reconstruction problem}.

The sequence reconstruction problem was studied in a variety of storage and communication scenarios \cite{chee2018coding, Cheraghchi.2019, Gabrys.2018, Konstantinova:2008, Levenshtein.2009, Sala.2017, yehezkeally2018reconstruction, Sini.2019, Junnila.2019}.
In these cases, $\C$ is usually assumed to be the entire space (all binary words of some fixed length) or a classical error-correcting code. 
However, in most storage scenarios, the number of noisy reads $N$ is a fixed system parameter and 
when $N$ is at most $\nu(\C;B)$, we are unable to {uniquely} reconstruct the codeword. 
In \cite{Kiah.arxiv.2020}, we propose the study of {\em code design} when the read coverage is strictly less than $\nu(\C;B)$.
Specifically, we say that $\C$ is an $(n,N;B)$-{\em reconstruction code} if 
$\C\subseteq \{0,1\}^n$ and $\nu(\C;B)<N$.

This gives rise to a {\em new quantity of interest} that measures the 
{\em trade-off between codebook redundancy and read coverage}. 
Specifically, given $N$ and an error-ball $B$, we study the quantity
%\vspace{-3mm}
%{\small
\begin{equation*}\label{eq:code.red}
\rho(n,N;B)\triangleq \min \Big\{n - \log|\C| : \C\subseteq \{0,1\}^n,\, \nu(\C;B) < N \Big\}.
\end{equation*}
%}

\subsection{The Sequence Reconstruction Problem for Deletion Channels}

{In this work, we focus on channels that {introduce} {\em deletions only}.
Specifically, let $\D_t(\bx)$ denote the deletion ball of $\bx$ with exactly $t$ deletions.
Let $D_t(n)$ denote the maximum deletion ball size of words of length $n$, that is, $D_t(n)=\max \{|\D_t(\bx)|: \bx\in\{0,1\}^n\}$. 
It is well known (see for example, \cite{Levenshtein.2001.jcta}) that 
\begin{equation}\label{eq:dtn1}
	D_t(n) = \sum_{i=0}^{t}\binom{n-t}{i}=n^{t}+O(n^{t-1}), \mbox{ for $0\le t\le n$}.
\end{equation}
For convenience, we assign $D_t(n)=0$ when $t<0$ or $t>n$.
% and it is achieved only for the alternating sequences $0101\cdots$ and $1010\cdots$~\cite{Hirschberg00}. 

For purposes of brevity, we let $\nu_t(n)$ denote $\nu(\{0,1\}^n;\D_t)$, the read coverage of $\{0,1\}^n$. 
We have the following landmark result of Levenshtein.
\begin{theorem}[Levenshtein \cite{Levenshtein.2001.jcta}]
%	Let $D_t(n)$ denote the maximum deletion ball size of words of length $n$. That is, $D_t(n)=\max \{|\D_t(\bx)|: \bx\in\{0,1\}^n\}$.
%	Then
	\begin{equation}\label{eq:read-wholespace}
	\nu_t(n)=2D_{t-1}(n-2)  = 2n^{t-1}+O(n^{t-2})\,.
	\end{equation}
%	Moreover, we have that
%	\begin{equation}\label{eq:ballsize}
%	D_t(n)=\sum_{i=0}^{t}\binom{n-t}{i}=n^{t}+O(n^{t-1})\, ,
%	\end{equation}
%	and therefore,
%	\begin{equation}\label{eq:read-wholespace-asymp}
%	\nu_t(n)=2n^{t}+O(n^{t-1})\,.
%	\end{equation}
\end{theorem}}

Recently, the authors of \cite{Gabrys.2018} studied the sequence reconstruction problem 
when $\C$ is a single-deletion-correcting code or an $(n,1;\D_1)$-reconstruction code.
Namely, they showed that $\C$ allows unique reconstruction with significantly less reads (as compared to $\nu_t(n)$) for deletions with $t\ge 2$.

\begin{theorem}[\hspace{-0.5pt}\cite{Gabrys.2018}]\label{thm:1-reconstruction}
	Let $\bx$ and $\by$ be two words of length $n\ge 7$. 
	If $\D_1(\bx)\cap \D_1(\by)=\varnothing$, then $|\D_t(\bx)\cap \D_t(\by)|\le N^{(1)}_t(n)$ for $t\ge 2$, where%
	%\footnote{\hm{HM: Should we parameterise using $n$? This is because the asymptotics for \eqref{eq:N1} is in $n$.}}
%	\begin{align} 
%	N^{(1)}_t(n) & = 2D_{t-2}(n-4) + 2D_{t-2}(n-5) + 2D_{t-2}(n-7) \notag\\
%		&\hspace{5mm} + D_{t-3}(n-6) + D_{t-3}(n-7)\notag \\
%		&=2n^{t-2}+O(n^{t-3}). \label{eq:N1}
%	\end{align}
	{\small
		\begin{align}
		N^{(1)}_t(n) & = 2D_{t-2}(n-4) + 2D_{t-2}(n-5) + 2D_{t-2}(n-7) \notag\\
		&\hspace{2mm} + D_{t-3}(n-6) + D_{t-3}(n-7)
		=2n^{t-2}+O(n^{t-3}). \label{eq:N1}
		\end{align}
		
	}
	Therefore, if $\C$ is an $(n,1;\D_1)$-reconstruction code,
	then $\nu(\C;\D_t)\le N^{(1)}_t(n)$ and so, $\C$ is also an $\left(n,N^{(1)}_t(n)+1;\D_t\right)$-reconstruction code for $t\ge 2$ and $n\ge 7$.
	Furthermore, this implies that $\rho\left(n,N^{(1)}_t(n)+1;\D_t\right)\le\log_2 n+O(1)$.
	%\footnote{\e{I didn't understand why this is true.}\hm{Sorry. My mistake. Hope this is clearer.}}.
\end{theorem} 

In the same spirit, we study the sequence reconstruction problem when the codebook $\C$ is an $(n,2;\D_1)$-reconstruction code.
Specifically, in Section~\ref{sec:2-reconstruct}, we show that 
if every channel introduces $t$ deletions, then it is possible to uniquely reconstruct
codewords from $\C$ with approximately $\nu_t(n)/2$ reads.
%we can uniquely reconstruct codewords from $\C$ with approximately $\nu_t(n)/2$ reads.

%\subsection{Reconstruction Codes with $N=2$ for Single-Deletion Channel}
\subsection{Reconstruction Codes with $N=2$ for Single Deletions}
%\vspace{-2mm}

We motivate the case for reconstruction codes in the context of the single-deletion channel. 
As mentioned early, when we use the whole space $\{0,1\}^n$ as our codebook, we require $\nu_1(n)=3$ noisy reads to uniquely reconstruct any codeword. Hence, we have $\rho(n,N;\D_1)=0$ for $N\ge 3$.

In contrast, when $N=1$, or, when we have only one noisy read, 
we recover the usual notion of error-correcting codes and 
the classical Varshamov-Tenengolts (VT) code is an $(n,1;\D_1)$-reconstruction code whose redundancy is at most $\log_2(n+1)$ ~\cite{Levenshtein.1966}. Hence, we have $\rho(n,1;\D_1)=\log_2 n+\Theta(1)$.
%{\color{red}Explain why $N=2$.}
Therefore, it remains to ask: how should we design the codebook when we have only two noisy reads? 
Or, what is the value of $\rho\left(n,2;\D_1\right)$? 

Now, the first construction of a $(n,2;\D_1)$-reconstruction code was proposed in \cite{chee2018coding} for the design of codes in racetrack memory. The codebook uses $\log_2\log_2 n +O(1)$ redundant bits and 
in \cite{Kiah.arxiv.2020}, we modified the construction to obtain codebooks that uniquely reconstruct codewords for the single-edit channel and its variants. 
The construction can be seen as a generalization of the classical Varshamov-Tenengolts (VT) code proposed by Levenshtein \cite{Levenshtein.1966} and the shifted VT codes proposed by Schoeny \etal{} \cite{Schoeny.2017}.

\begin{definition}[Constrained Shifted VT Codes \cite{chee2018coding, Kiah.arxiv.2020}]\label{def:csvt}
	For $n\ge P >0$ and $P$ even, let $c\in \bbZ_{1+P/2}$ and $d\in \bbZ_2$.
	The {\em constrained shifted VT code} $\CSVT(n,P;c,d)$ is defined to be the set of all words $\bx=x_1x_2\cdots x_n$ 
	such that the following holds. 
	\begin{enumerate}[(i)]
		\item $\syn(\bx) = c \pmod{1+P/2}$.
		\item $\sum_{i=1}^n x_i = d\pmod{2}$. 
		%\item The longest run of zeroes, run of ones and alternating run in $\bx$ are at most $P$.
		\item The longest 2-periodic run in $\bx$ is at most $P$.
	\end{enumerate}
	Here, $\syn(\bx)$ denotes the {\em VT syndrome} $\syn(\bx)\triangleq \sum_{i=1}^n ix_i$ and
	a {\em 2-periodic} run refers to a continguous substring $x_ix_{i+1}\cdots x_j$ where $x_k=x_{k+2}$ for all $i\le k\le j-2$. 
	%an {\em alternating} run refers to a continguous substring $x_ix_{i+1}\cdots x_j$ where $x_i\ne x_{i+1}$ and $x_k=x_{k+2}$ for all $i\le k\le j-2$.
\end{definition}

When $P=2n$ and we remove Condition (ii)%
\footnote{When $P=2n$, then any 2-periodic run is at most $n<P$. Hence, Condition (iii) is always true.} 
	%\e{what does it mean?} \hm{Hope it is clearer now}}, 
we recover the classical VT code that corrects a single deletion.
On the other hand, when we remove the Condition (iii), we recover the shifted VT code that 
is used in the correction of a single burst of deletions \cite{Schoeny.2017}.
It was recently demonstrated that the CSVT code enables unique reconstruction whenever we have two distinct noisy reads.

\begin{theorem}[\hspace{-0.5pt}\cite{chee2018coding, Kiah.arxiv.2020}]\label{thm:code-t1-N2}
	For all choices of $c$ and $d$, we have that  $\CSVT(n,P;c,d)$ is an $(n, 2; \D_1)$-reconstruction code.
	Furthermore, if we set $P=\ceil{\log_2 n}+2$, the code $\CSVT(n,P;c,d)$ has redundancy $1+\log_2(\ceil{\log+2 n}+4)=\log_2 \log_2 n+O(1)$ 
	for some choice of $c$ and $d$.
	Thus, $\rho(n,2;\D_1)\le \log_2\log_2 n +O(1)$.
\end{theorem}

In this paper, we demonstrate that the codes in Theorem~\ref{thm:code-t1-N2} are asymptotically {\em optimal}. 
Specifically, in Section~\ref{sec:upper}, we show that an $(n, 2; \D_1)$-reconstruction code requires at least $\log_2\log_2 n - O(1)$ redundant bits. 

To demonstrate this necessary condition, we first observe that $\nu_1(n)=2$ and thus, we need to characterize pairs of words whose single-deletion balls have intersection size exactly two. 
To do so, we have the following definition of confusability.

\begin{definition}\label{def:type-A-confusable}
Two words $\bx$ and $\by$ are {\em Type-A-confusable} if 
\[\bx=\bu\ba\bv, \mbox{ and } 
\by=\bu\bara\bv,%\footnote{\e{should it be here $\by=\bu\ba\bv$?}}
\]
%\begin{align*}
%\bx&=\bu\bara\bv, \\ 
%\by&=\bu\bara\bv,
%\end{align*}
for some subwords $\ba$, $\bu$, and $\bv$ such that 
$|\ba|\ge 2$, $\bara$ is the complement of $\ba$,
and $\ba=a_1a_2\ldots a_j$ is an {\em alternating sequence}, 
that is, $\ba$ is 2-periodic and $a_1\ne a_2$. 
%{\color{red}is this $a_1 \neq a_2$ necessary?, and have mentioned why we consider only $|\D_1(\bx)\cap \D_1(\by)|=1$ or $2$?, i.e. because $\nu_1(n)=2$} \hm{HM: yes, we need $a_1\ne a_2$ for Lemma~\ref{char:type-A} to be correct.}
\end{definition}

The following characterization was demonstrated in \cite{Kiah.arxiv.2020}.

\begin{lemma}[Type-A-confusability~\cite{Kiah.arxiv.2020}]\label{char:type-A}
	Let $\bx$ and $\by$ be binary words.
	We have that $|\D_1(\bx)\cap \D_1(\by)|= 2$ if and only if 
%	the Hamming distance of $\bx$ and $\by$ is at least two and 
	$\bx$ and $\by$ are Type-A-confusable.
\end{lemma}

In Section~\ref{sec:2-reconstruct}, we derive an analogous result that characterizes when two single-deletion balls intersect at exactly one word. Using this characterization, we then analyse the read coverage of an $(n,2;\D_1)$-reconstruction code.

\subsection{Main Contributions}

In summary, our contributions are as follows.

\begin{itemize}
\item In Section~\ref{sec:upper}, we consider the case where $t=1$ and $N=2$,
and  demonstrate that a $(n,2;\D_1)$-reconstruction code requires at least $\log_2\log_2 n-O(1)$ bits of redundancy.
Therefore, the CSVT code constructed in Theorem~\ref{thm:code-t1-N2} is asymptotically optimal and 
we have that $\rho(n,2;\D_1)=\log\log n+\Theta(1)$.
Furthermore, we have the complete solution for $\rho$ in the case for $t=1$.
\vspace{2mm}

\begin{theorem}\label{thm:complete-solution}
The value $\rho(n,N;\D_1)$ satisfies
%	Let $\D_1$ denote the single-deletion ball. Then 
	\[ \rho(n,N;\D_1)=
	\begin{cases}
	\log_2 n + \Theta(1), & \text{when } N = 1,\\
	\log_2 \log_2 n + \Theta(1), & \text{when } N = 2,\\
	0, & \text{when } N\ge 3.
	\end{cases}\]
\end{theorem}
Theorem~\ref{thm:complete-solution} shows that as the number of noisy reads increases, the optimal number of redundant bits required is gracefully reduced from $\log_2 n + \Theta(1)$ to $\log_2\log_2 n + \Theta(1)$, and then to zero.
\item In Section~\ref{sec:2-reconstruct}, we consider the case $t\ge 2$ and we show that if $|\D_1(\bx)\cap \D_1(\by)|=1$, then $|\D_t(\bx)\cap \D_t(\by)| \le D_{t-1}(n-1)+\nu_{t-2}(n-3)$. 
Hence, for the special case of $t=2$, an $(n,2;\D_1)$-reconstruction code can uniquely reconstruct codewords with $n+1$ distinct reads. 
By refining our arguments, we show that with appropriate choice of $P$, 
the constrained SVT codes from Theorem~\ref{thm:code-t1-N2} can uniquely reconstruct codewords with strictly less than $n+1$ distinct reads.
\end{itemize}

\section{Lower Bound for $t=1$ and $N=2$}
\label{sec:upper}

In this section, we provide a lower bound on the number of redundant bits of an $(n,2;\D_1)$-reconstruction code $\C$,
or equivalently, an upper bound on the size of $\C$. 
To this end, we borrow graph theoretic tools and consider the graph $\G(n)$ whose vertices correspond to $\{0,1\}^n$.
The vertices $\bx$ and $\by$ are adjacent if and only if $|\D_1(\bx)\cap \D_1(\by)|=2$, or equivalently, 
$\bx$ and $\by$ are Type-A-confusable.

Hence, $\C$ is an $(n,2;\D_1)$-reconstruction code if and only if the
corresponding set of vertices are independent in $\G(n)$.

\begin{definition}A collection $\Q$ of cliques is a {\em clique cover} of $\G$ if 
every vertex in $\G$ belongs to some clique in $\Q$.
\end{definition}

We have the following fact from graph theory 
(see for example, \cite{Knuth.1994}).

\begin{theorem}
	If $\Q$ is a clique cover, then the size of any independent set is at most $|\Q|$.
\end{theorem}
Therefore, our objective is to construct a clique cover for $\G(n)$.
To this end, we consider another parameter $\ell$,
and set $m=\floor{n/(2\ell)}$ and $r=n-2\ell m$.
We divide each word of length $n$ into $m$ blocks of length $2\ell$ and one block of length $r$.

Set 
\[\small
\Lambda = \Big\{(01)^j(10)^{\ell-j}: j\in [\ell]\Big\}~\cup~
\Big\{(10)^j(01)^{\ell-j}: j\in [\ell]\Big\}
\]
and 
$\widetilde{\Lambda} = \{0,1\}^{2\ell} \setminus\Lambda$. 
So, $|\Lambda|=2\ell$ and $|\widetilde{\Lambda}|=2^{2\ell}-2\ell$.
To construct our clique cover $\Q(n,\ell)$, we consider two types of cliques.
The first type of cliques are singletons of the form
%\[ S_\bx = \{\bx\}, \mbox{ where } \bx=\bx_1\bx_2\cdots \bx_m, \, \bx_i\notin \Lambda \mbox{ for all } i\in [m]. \]
\[ S_\bx = \{\bx\}, \mbox{ where } \bx \in \widetilde{\Lambda}^m \times \{0,1\}^r . \]
The second type of cliques are cliques of size $\ell$. 
Here, we define
%\[ \Gamma = \{(\bx_1\bx_2\cdots \bx_{i-1}, \by_{i+1}\by_{i+2}\by_{m}, i ): \bx_j\notin \Gamma, \by_j\in \{0,1\}^{2\ell}, i\in[m]\}.\]
\[ \Gamma = \left\{(\bu, \bw, i ): \bu\in \widetilde{\Lambda}^{i-1},\, \bw \in \{0,1\}^{2\ell(m-i+)+r},\, i\in[m]\right\}.\]
For each $\bz=(\bu,\bw,i)$, we define two sets of vertices (which we later show to be cliques of size $\ell$):
\begin{align*}
Q_\bz^{(0)} & =\left\{\bu(01)^j(10)^{\ell-j} \bw: j\in [\ell]\right\},\\
Q_\bz^{(1)} & =\left\{\bu(10)^j(01)^{\ell-j} \bw: j\in [\ell]\right\}.
\end{align*}
We then define 
\[ \Q(n,\ell) = \left\{S_\bx: \bx \in \widetilde{\Lambda}^m\times \{0,1\}^r\right\} 
\cup \left\{Q_\bz^{(0)}, Q_\bz^{(1)}: \bz\in \Gamma\right\}.\]

\begin{lemma}$\Q(n,\ell)$ is a clique cover for $\G(n)$.
\end{lemma}

\begin{proof}
	Clearly, all singletons are cliques. 
	Next, we show that the $\ell$-set $Q^{(\mu)}_\bz$ is a clique for all $\bz\in\Gamma$ and $\mu\in\{0,1\}$.
	We assume $\mu=0$ and the proof for $\mu=1$ is similar.
	
	Let $\bx=\bu(01)^i(10)^{\ell-i}\bw$ and $\by=\bu(01)^j(10)^{\ell-j}\bw$ be two words in $Q^{(0)}_\bz$.
	Without loss of generality, let $i<j$. Then we can rewrite $\bx$ and $\by$ as
	\begin{align*}
	\bx &  = \bu(01)^i(10)^{j-i}(10)^{\ell-j}\bw, \\
	\by &  = \bu(01)^i(01)^{j-i}(10)^{\ell-j}\bw.
	\end{align*}
	Thus, $\bx$ and $\by$ are Type-A-confusable 
	%Then $\bx$ and $\by$ are Type-A-confusable with the prefix $\bu(01)^i$ and the suffix $(10)^{\ell-j}\bw$%
	%\footnote{\e{I don't think that this sentence is clear.}\hm{Hope it is clearer now.}}. 
	and so, $\bx$ and $\by$ are adjacent in $\G(n)$.
	Therefore, $Q^{(0)}_\bz$ is a clique.
	
	It remains to show that any word $\bx\in\{0,1\}^n$ belongs to some clique in $\Q(n,\ell)$.
	If $\bx\in\widetilde{\Lambda}^m$, then $\bx\in S_\bx$.
	Otherwise, $\bx\not\in\widetilde{\Lambda}^m$ and one of the $m$ subblocks of $\bx$ belongs to $\Lambda$.
	Let the $i$th subblock be the first subblock from the left that belongs to $\Lambda$. 
	Hence, this subblock is either of the form $(01)^j(10)^{\ell-j}$ or $(10)^j(01)^{\ell-j}$ for some $j\in [\ell]$.
	In the first case, $\bx$ belongs to $Q^{(0)}_{(\bu,\bw,i)}$ where $\bu$ is the first $(i-1)$ subblocks and $\bw$ is the last $(m-i+1)$ subblocks. 
	In the second case, $\bx$ belongs to $Q^{(1)}_{(\bu,\bw,i)}$ where $\bu$ and $\bw$ are similarly defined.
\end{proof}

\begin{example}\label{exa:ell2}
Set $\ell=2$ and so, $\Lambda = \{0110,0101,1001,1010\}$.
When $m=3$, a possible element $\bz$ in $\Gamma$ is the triple $(0000,1000,1)$ and 
the cliques corresponding to $\bz$ are
\begin{align*}
Q_\bz^{(0)} & =\left\{000001101000,000001011000\right\},\\
Q_\bz^{(1)} & =\left\{000010011000,000010101000\right\}.
\end{align*}
For general $m$, since $|\widetilde{\Lambda}|=12$, the number of singletons is $12^m$.
Furthermore, the number of $\ell$-cliques is $2|\Gamma|$. 
Since the size of $\Gamma$ is given by $\sum_{i=1}^m 12^{i-1} 2^{4(m-i)}=2^{n-2}(1-(3/4)^m)$,
we have that the size of the clique cover $\Q(n,2)$ is
\[2\cdot \left(2^{n-2}(1-(3/4)^m)\right)+12^m = 2^{n-1}(1+o(1)).\]
Therefore, $\log |\Q(n,2)| = n-1+o(1)$.
Thus, an $(n,2;\D_1)$-reconstruction code requires at least one redundant bit asymptotically. \qed
\end{example}

To obtain the lower bound of $\log_2\log_2 n - o(1)$ redundant bits, 
we refine our analysis by allowing $\ell$ to grow with $n$.

Now, we write $\lambda=|\widetilde{\Lambda}|=2^{2\ell}-2\ell$. 
Similar to the analysis in Example~\ref{exa:ell2}, we have the following lemma.

\begin{lemma}
	The size of $\Q(n,\ell)$ is given by 
	\[2^n \left\{\left(1-\frac{2\ell}{2^{2\ell}}\right)^{\floor{\frac{n}{2\ell}}}+\frac{1}{\ell}\left( 1-\left(1-\frac{2\ell}{2^{2\ell}}\right)^{\floor{\frac{n}{2\ell}}} \right)   \right\}.\]
\end{lemma} %is \footnote{\e{should we have this part as a lemma/theorem statement with a proof?}}

\begin{proof}
	Recall that $m=\floor{n/(2\ell)}$. The number of singletons is $\lambda^m$, while the number of $\ell$-cliques is $2|\Gamma|$, 
	where $|\Gamma|=\sum_{i=1}^m \lambda^{i-1} 2^{2\ell(m-i)}$. Hence, the size of  $\Q(n,\ell)$ is
	{\small
		\[
		\lambda^{m} + 2 \sum_{i=1}^{m}{\lambda^{i-1} 2^{2\ell(m-i)}}
		= \lambda^{m} + 2^{n-2\ell+1}\frac{(\frac{\lambda}{2^{2\ell}})^{m}-1}{\frac{\lambda}{2^{2\ell}}-1}.
		\]}
	%	\begin{align*}
	%	& \lambda^{m} + 2 \sum_{i=1}^{m}{\lambda^{i-1} 2^{2\ell(m-i)}} \\
	%	&= \lambda^{n/(2\ell)} + 2^{n-2\ell+1}\frac{({\frac{\lambda}{2^{2\ell}})^{{n}/(2\ell)}}-1}{\frac{\lambda}{2^{2\ell}}-1}.%\\
	%	%&=2^n \left\{\left(1-\frac{2\ell}{2^{2\ell}}\right)^{\frac{n}{2\ell}}+\frac{1}{\ell}\left( 1-\left(1-\frac{2\ell}{2^{2\ell}}\right)^{\frac{n}{2\ell}} \right)   \right\}. \qedhere
	%	\end{align*}
	Straightforward manipulations then yield the lemma.
\end{proof}
We set $\ell = \floor{\frac{1}{2} \left( 1-\epsilon\right) \log_2 n}$  where $0 <\epsilon <1$ 
and write $f(n)=\left(1-\frac{2\ell}{2^{2\ell}}\right)^{\frac{n}{2\ell}}$.
Hence,
\begin{align*}
\log_2 |\Q(n,\ell)| &= n - \log_2 \ell + \log_2 (1+(\ell -1)f(n))\\
&\leq n - \log_2 \ell + \log_2 (1+\ell f(n)).
\end{align*}

Since $\log_2 \ell \geq \log_2\log_2 n - O(1)$, it suffices to show that $\log_2 (1+\ell f(n))=o(1)$. %tends to zero as $n$ grows.
% show the following lemma.

\begin{lemma}\label{claim:upperbound}
	We have that $\lim_{n \to \infty}{\ell f(n)}=0$, or equivalently, $\lim_{n\to\infty} \ln(\ell f(n))=-\infty$. 
\end{lemma}

\begin{proof}
	First, we show that 
	\begin{equation}
	\lim_{n \to \infty}{\frac{\ln \ell}{\ln f(n)}}=0. \label{eq:f} 
	\end{equation}
	Note that for $0<x<1$, we have $|\ln (1-x)| \geq x$.
	Therefore,
	\begin{align*}
	\lim_{n \to \infty}{\left|\frac{\ln \ell}{\ln f(n)}\right|}&=\lim_{n \to \infty}{\left|\frac{\ln \ell}{\floor{\frac{n}{2\ell}}  \ln\left( 1-\frac{2\ell}{2^{2\ell}}\right)}\right|} \\
	&\leq \lim_{n \to \infty}{\left|\frac{\ln \ell}{\left( \frac{n}{2\ell}-1  \right)  \frac{2\ell}{2^{2\ell}}}\right|} \\
	&=\lim_{n \to \infty}{\left|\frac{2^{2\ell}\ln \ell}{n-2 \ell}\right|} \\
	&\leq \lim_{n \to \infty}{\left|\frac{2^{2+\left( 1-\epsilon\right) \log_2 n}\ln \ell}{n-2 \ell}\right|}\\
	&=4 \lim_{n \to \infty}{\left|\frac{n \ln \ell}{n^{\epsilon}\left(n-2 \ell\right)}\right|}\\
	&=4  \lim_{n \to \infty}{\left|\frac{n}{n-2 \ell}\right|}\lim_{n \to \infty}{\left|\frac{\ln \ell}{n^{\epsilon}}\right|}\\
	&=4 \times 1 \times 0=0,
	\end{align*}
	which implies \eqref{eq:f}. Note that since $\lim_{n \to \infty}{\ln \ell}=\infty$, and $f(n)<1$ for sufficiently large $n$, combined with \eqref{eq:f}, this implies that $\lim_{n \to \infty}{\ln f(n)}=-\infty$. Therefore, together with \eqref{eq:f}, we have the following:
	\begin{align*}
	\lim_{n \to \infty}{\ln (\ell f(n)) }&=\lim_{n \to \infty}{\ln \ell + \ln f(n)} \\
	&= \lim_{n \to \infty}{\left(\ln f(n)\right) \left(1+\frac{\ln \ell}{\ln f(n)}\right)}  \\
	&= \lim_{n \to \infty}{\ln f(n)} \lim_{n \to \infty} {\left(1+\frac{\ln \ell}{\ln f(n)}\right)} \\
	&=\lim_{n \to \infty}{\ln f(n)}=-\infty. \qedhere
	\end{align*}
\end{proof}

Therefore, the results in this section can be summarized in following theorem.
\begin{theorem}
Let $\C$ be an $(n,2;\D_1)$-reconstruction code.
For $\epsilon>0$, we have that
\begin{equation}
\log_2 |\C| \le n - \log_2\log_2 n + \log_2(1-\epsilon)+o(1).
\end{equation}
Therefore, $\rho(n,2;\D_1) = \log_2\log_2 n - O(1)$.
Combining with Theorem~\ref{thm:code-t1-N2}, we have that $\rho(n,2;\D_1) = \log_2\log_2 n +\Theta(1)$.
\end{theorem}

\section{Reconstruction Codes for $t\ge 2$ Deletions}
\label{sec:2-reconstruct}

In this section, we demonstrate the following result.

\begin{theorem}\label{thm:2-reconstruction}
	Let $\bx$ and $\by$ be binary words of length $n\ge 6$ and $t\ge 2$.
	If $|\D_1(\bx)\cap\D_1(\by)|=1$, then we have that 
	\begin{align}
	|\D_t(\bx)\cap\D_t(\by)| 
	& 	\le D_{t-1}(n-1)+\nu_{t-1}(n-3) \label{eq:coverage2} \\
	& =n^{t-1}+O(n^{t-2}) \text{ for fixed values of } t. \notag
	\end{align}
	Furthermore, when $t< n/2$, the inequality is strict.
\end{theorem}

%\begin{theorem}\label{thm:2-reconstruction}
%	Let $\bx$ and $\by$ be binary words of length $n\ge 6$
%	and $t\ge 2$.
%	If $\D_1(\bx)\cap\D_1(\by)=\{\bz\}$, then we have that 
%	\begin{equation}\label{eq:coverage1}
%	|\D_t(\bx)\cap\D_t(\by)|\le |\D_{t-1}(\bz)|+\nu_{t-1}(n-3).
%	\end{equation}
%	Hence, for fixed values of $t$,
%	\begin{equation}\label{eq:coverage2}
%	\small
%	|\D_t(\bx)\cap\D_t(\by)|\le D_{t-1}(n-1)+\nu_{t-1}(n-3)=n^{t-1}+O(n^{t-2}).
%	\end{equation}
%	Furthermore, when $t\le n/2$, the inequality is strict.
%\end{theorem}

Before we provide the detailed proof of Theorem~\ref{thm:2-reconstruction}, we look at its implication.
Suppose that we have an $(n,2;\D_1)$-reconstruction code $\C$. 
%Then the single-deletion balls of any two codewords in $\C$ intersect at most one word.
Then the intersection size of the single-deletion balls of any two codewords in $\C$ is at most one.
Applying Theorem~\ref{thm:2-reconstruction}, we have that the read coverage $\nu(\C;\D_t)$  is at most $N_t^{(2)}(n)$ where $N_t^{(2)}(n) = D_{t-1}(n-1)+\nu_{t-1}(n-3)$.
Hence, $\C$ is an $(n,N_t^{(2)}(n)+1;\D_t)$-reconstruction code.
We also observe that $N_2(n)\sim \nu_t(n)/2$, or, 
$\lim_{n\to \infty} N_t^{(2)}(n)/\nu_t(n) =1/2$.
Therefore, by sacrificing $\log_2\log_2 n+O(1)$ bits of information, the codes in Theorem~\ref{thm:code-t1-N2} are able to uniquely reconstruct codewords with half the number of noisy reads (as compared to no coding).
Note also that by Theorem~\ref{thm:1-reconstruction}, if the number of redundancy is roughly $\log_2 n$, then the number of noisy reads has to be $2n^{t-2}+O(n^{t-3})$.
We summarize our discussion with the following theorem.

\begin{theorem}
	Let $n\ge 6$ and $t\ge 2$.
	Set $N_t^{(2)}(n)=D_{t-1}(n-1)+\nu_{t-1}(n-3)$.
	If $\C$ is an $(n,2;\D_1)$-reconstruction code,
	then $\C$ is also an $\left(n,N_t^{(2)}(n)+1;\D_t\right)$-reconstruction code.
	Furthermore, this implies that $\rho\left(n,N_t^{(2)}(n)+1;\D_t\right)\le  \log_2\log_2 n+O(1)$.
\end{theorem}

\begin{remark}
	When $t=2$, we have that $N_2^{(2)}(n)=n+1$. In Section~\ref{sec:t2}, we focus on this special case and show that constrained SVT codes in Definition~\ref{def:csvt} are able to uniquely reconstruct codewords with strictly less than $n+1$ reads. 
\end{remark}

As the proof of Theorem~\ref{thm:2-reconstruction} is fairly technical, we outline our proof strategy.
\begin{itemize}
	\item First, in Section~\ref{sec:characterization}, we provide a characterization lemma similar to Lemma~\ref{char:type-A}.
	Specifically, we describe the necessary conditions for a pair of words to have single-deletion balls intersecting at exactly one output word.
	\item Applying the characterization lemma, we consider pairs of words with certain properties. 
	In Section~\ref{sec:firsttwocases}, we analyse the intersection size of certain $t$-deletion balls under certain scenarios. %To do so, we apply the characterization lemma.
	\item Finally, in Section~\ref{sec:completeproof}, we use an inductive argument to complete the proof. 
\end{itemize}
%To prove our main theorem for the section, we first provide a characterization lemma similar to Lemma~\ref{char:type-A}.
%Using this characterization, we analyse the intersection size of certain $t$-deletion balls under certain scenarios.
%After which, we provide a proof of \eqref{eq:coverage2}.
%In Section~\ref{sec:strict}, we demonstrate the inequality \eqref{eq:coverage2} is strict whenever $t\le n/2$.

%{\color{blue}
%	Johan: please do the following:
%	\begin{itemize}
%		\item Check the statements of the lemma.
%		\item Provide proof when necessary.
%		\item Do take note of the notation:
%		\begin{itemize}
%			\item I use {\tt mathcal} ${\cal D}_t(\bx)$ to indicate the ball and $D_t(n)$ to indicate the maximum ball size.
%			
%			\item I use $\nu_t(n)$ instead of $N_t(n)$.
%			
%			\item Type-B-confusability refers to Definition~\ref{def:type-B-confusable}. Here, only the characteristics of the word matter and Type-B-confusable words may be Type-A-confusable.
%			
%			\item Use words instead of sequences.
%			 
%		\end{itemize}
%		\item Refrain from using $\mathbf{U_n^{(j)}}$ until Section~\ref{sec:numruns}.
%		\item Highlight all changes in BLUE. Feel free to point out any mistakes or inconsistency in the other sections.
%	\end{itemize}
%}

\subsection{Type-B-Confusability}\label{sec:characterization}

To characterize words whose single-deletion balls intersect at exactly one word, we introduce the following notion of confusability.

\begin{definition}\label{def:type-B-confusable}
	Two words $\bx$ and $\by$ are {\em Type-B-confusable} if 
	\[
	\bx=\bu a \overline{a} \bv b \bw \quad \text{and}\quad \by=\bu \overline{a} \bv b\overline{b}\bw,
	\] or vice versa,
	for some subwords $\bu$, $\bv$ and $\bw$, and $a,b \in \{0,1\}$.
\end{definition}

Next, we borrow certain notation from \cite{Gabrys.2018}.	
Let $\cS$ be a set of binary words and $a,b \in \{0,1\}$. 
We define $\cS^a$ to be the set of all words in $\cS$ that start with $a$ and $\cS_b$ to be the set of all words in $\cS$ that end with $b$. 
We also combine both notations and let $\cS_b^a$ be the set of all words in $\cS$ that start with $a$ and end with $b$. 
If $\bx$ is a word, we define $\cS \circ \bx$ (or $\bx \circ \cS$) to be set of all words obtained by appending (or prepending) $\bx$ to every word in $\cS$.

%The following technical lemma is useful in our proofs.
%
%\begin{lemma}[Gabrys and Yaakobi \cite{Gabrys.2018}]
%	Let $a\in\{0,1\}$ and $\bv$ be a word. Then we have the following.
%	\begin{align*}
%	\D_t(a\bv)^a & = a\circ\D_t(\bv),\\
%	\D_t(\overline{a}\bv)^a & = \D_{t-1}(\bv)^a.
%	\end{align*}	
%\end{lemma}

\begin{lemma}\label{lemma:lev1form}
	Let $\bx$ and $\by$ be two binary words. 
	If $|\D_1(\bx)\cap\D_1(\by)|=1$, then 
	either the Hamming distance of $\bx$ and $\by$ is one or $\bx$ and $\by$ are Type-B-confusable.
\end{lemma}

\begin{proof}
	 Suppose that $\bx$ and $\by$ have Hamming distance at least two. 
	 Then $\bx$ and $\by$ must be of the form 
	 \[\bu a \bd b \bw \text{ and } \bu \overline{a} \be \overline{b} \bw\] 
	 for subwords $\bu$, $\bw$, $\bd$, and $\be$, where $|\bd|=|\be|$, and $a,b \in \{0,1\}$.
	
	Without loss of generality, suppose that $\bx=\bu a \bd b \bw$ and $\by=\bu \overline{a} \be \overline{b} \bw$. If $\bd$ is empty, then $\bx=\bu a b \bw$ and $\by=\bu \overline{ab} \bw$. If $a =b$, then their weight differ by two and hence $|{\cal D}_1(\bx) \cap {\cal D}_1(\by)|= 0$ which contradicts our assumption. Else, if $a \neq b$, then by definition we have that $\bx$ and $\by$ are \textit{Type-A-confusable}, and by Lemma \ref{char:type-A}, we have $|{\cal D}_1(\bx) \cap {\cal D}_1(\by)|= 2$  which also contradicts our assumption. Therefore $\bd$ is nonempty.
	
	Let $\{\bz\} = {\cal D}_1(\bx) \cap {\cal D}_1(\by)$. Note that the following intersection of 1-deletion balls are empty:
	\begin{align*}
	&{\cal D}_1(\bu a)\ \circ \bd b\bw \cap {\cal D}_1(\bu \overline{a}) \circ \be \overline{b} \bw = \varnothing \\
	&\bu a\ \circ {\cal D}_1(\bd b\bw) \cap \bu \overline{a} \circ {\cal D}_1(\be \overline{b} \bw) = \varnothing \\
	&{\cal D}_1(\bu a \bd) \circ b\bw \cap {\cal D}_1(\bu \overline{a} \be) \circ \overline{b} \bw=\varnothing,\\
	&\bu a \bd \circ {\cal D}_1(b\bw) \cap \bu \overline{a} \be \circ {\cal D}_1(\overline{b} \bw)=\varnothing
	\end{align*}Hence $\bz$ can only be in ${\cal D}_1(\bu a) \circ \bd b\bw \cap \bu \overline{a} \be \circ {\cal D}_1(\overline{b} \bw)$ or $\bu a \bd \circ {\cal D}_1(b\bw) \cap {\cal D}_1(\bu \overline{a}) \circ \be \overline{b} \bw$. Without loss of generality, we assume that $\bz \in {\cal D}_1(\bu a) \circ \bd b\bw \cap \bu \overline{a} \be \circ {\cal D}_1(\overline{b} \bw)$. 
	Matching positions implies that $\bu \in {\cal D}_1(\bu a), \bd b = \overline{a} \be$ and $\bw \in {\cal D}_1(\overline{b} \bw)$. Furthermore it implies that $\bz=\bu \bd b \bw$. Let $\bd = \overline{a} \bv$ and $\be=\br b$ for some subwords $\bv$ and $\br$. Since $\bd b = \overline{a} \be$, we have $\overline{a} \bv b = \overline{a} \br b$, and hence $\bv=\br$.
	
	Therefore we have shown that $\bx=\bu a \bd b \bw=\bu a \overline{a} \bv b \bw$ and $\by=\bu \overline{a} \be \overline{b} \bw=\bu \overline{a} \bv b \overline{b}\bw$.
\end{proof}

\subsection{Special Cases}\label{sec:firsttwocases}

Following Lemma~\ref{lemma:lev1form}, we study the intersection size of $t$-deletion balls for two special cases.
In the first case, we assume that the two words differ at exactly one coordinate.
In the second case, we assume that the words are Type-B-confusable with $\bu$ and $\bw$ being empty strings.

In our proofs, we appeal to the following technical results on deletion balls.

\begin{lemma}[\hspace{-0.3pt}\cite{Levenshtein.2001.jcta,Liron.2015, Gabrys.2018}]\label{lemma:dtn}
	Let $1\le t\le n$ and $a\in\{0,1\}$. Suppose that $\bu$, $\bv$, $\bx$, and $\by$ are binary words. 
	\begin{enumerate}[(i)]
		\item In addition to \eqref{eq:dtn1} and \eqref{eq:read-wholespace}, 
		we have that 
		\begin{align*}
		D_t(n) &= D_{t}(n-1)+D_{t-1}(n-2),\\
		\nu_t(n) &=\nu_{t}(n-1)+\nu_{t-1}(n-2).
		\end{align*}
		\item $D_t(n) \geq D_{t-i}(n-i)$ and 
		$\nu_t(n) \geq \nu_{t-i}(n-i)$ for $ i \leq t$.
		%\item If $\bx=\bu aa\bv$ and $\by= \bu a\overline{a\bv}$, then $|\D_t(\bx)| \leq |\D_t(\by)|$.
		%\item Let $\bx=\bu aa\bv$ where $\bu,\bv$ are subwords and $a \in \{0,1\}$. If $\by= \bu a\overline{a\bv}$, then $|\D_t(\bx)| \leq |\D_t(\by)|$.
		\item $\D_t(a\bx)^a  = a\circ\D_t(\bx)$ and $\D_t(\overline{a}\bx)^a = \D_{t-1}(\bx)^a$.
		\item $|D_t(\bx)^a|\le D_t(|\bx|-1)$.
		%	\D_t(\overline{a}\bv)^a & = \D_{t-1}(\bv)^a.
		%\item Let $a\in\{0,1\}$ and $\bv$ be a word. Then we have the following.
		%	\begin{align*}
		%	\D_t(a\bv)^a & = a\circ\D_t(\bv),\\
		%	\D_t(\overline{a}\bv)^a & = \D_{t-1}(\bv)^a.
		%	\end{align*}
		\item Suppose further that $t< n/2$. Then $|\D_t(\bx)|=D_t(n)$ if and only if $\bx$ is an alternating sequence.
	\end{enumerate}
\end{lemma}

%We also require the following lemma.	
%\begin{lemma}\label{lemma:startswithaletter}
%	Let $a\in\{0,1\}$ and $\bv$ be a word. Then we have that
%	\begin{equation}\label{eq:n-1}
%	|D_t(\bv)^a|\le D_t(|\bv|-1).
%	\end{equation}
%\end{lemma}

\begin{proof}
	(i) and (ii) are from Levenshtein's work \cite{Levenshtein.2001.jcta}, while 
	%(iii) and (iv) are derived in \cite{Liron.2015} and \cite{Gabrys.2018}, respectively.
	(iii) is derived in \cite{Gabrys.2018}.
	
	We prove (iv) here. If $a$ does not appear in $\bx$, then the inequality is trivial. 
	Now suppose that $\bx={\overline{a}^m}a\bx^*,$ for some $m \geq 0$, and subword $\bx^*$. 
	Then we have $|D_t(\bx)^a|=|D_t({\overline{a}^m}a\bx^*)^a|=|D_{t-m}({\bx^*})|\leq D_{t-m}(|\bx^*|)=D_{t-m}(|\bx|-m-1) \leq D_{t}(|\bx|-1)$, where the last inequality follows from Lemma~\ref{lemma:dtn}(ii).
	
	Next, we prove (v). When $\bx$ is alternating, it is straightforward to verify that $|\D_t(\bx)|=D_t(n)$.
	To show the converse, we suppose that $\bx$ is not alternating. 
	Then \cite[Claim 6]{Gabrys.2018} states that $|\D_t(\bx)|\le D_t(n-2)+D_{t-1}(n-2)+D_{t-2}(n-4)$.
	Applying Lemma~\ref{lemma:dtn}(i), we have that $|\D_t(\bx)|<D_t(n)$ if $D_{t-2}(n-4) < D_{t-1}(n-3)$.
	Now, since the difference $D_{t-1}(n-3)-D_{t-2}(n-4)=\binom{n-t-2}{t-1}$, 
	we have a strict inequality when $n-t-2\ge t-1$, or, $t<n/2$.
\end{proof}

We proceed to study the first special case where $\bx$ and $\by$ have Hamming distance one. 

\begin{lemma}\label{lemma:hamming1forgeneral}
	Let $\bx$ and $\by$ be words with Hamming distance one.
	That is, $\bx=\bu 1 \bv$ and $\by=\bu 0 \bv$ for subwords $\bu$ and $\bv$.
	Then $\D_t(\bx) \cap \D_t(\by)=\D_{t-1}(\bu \bv) $ for any $t\geq 1$.
\end{lemma}

\begin{proof}
	We first show the result for $t=1$. 
	%i.e. if $\bx=\bu 1 \bv$ and $\by=\bu 0 \bv$, then 
	i.e. ${\cal D}_1(\bx) \cap {\D}_1(\by)={\D}_{0}(\bu \bv) =\{\bu \bv\}$. 
	Note that $\bu\bv \in {\cal D}_1(\bx) \cap {\cal D}_1(\by)$. Suppose there exists $\bz \in {\cal D}_1(\bx) \cap {\cal D}_1(\by)$ where $\bz \neq \bu\bv$. Then we must have $\bz \in {\cal D}_1(\bu)\circ 1\bv \cap \bu 0\circ {\cal D}_1(\bv)$ or $\bz \in \bu 1\circ{\cal D}_1(\bv) \cap {\cal D}_1(\bu)\circ 0\bv$. Without loss of generality, suppose that $\bz \in {\cal D}_1(\bu)\circ 1\bv \cap \bu 0\circ {\cal D}_1(\bv)$. By matching positions, we must have $\bu = \bu^* 1$ and $\bv=0\bv^*$ for some subwords $\bu^*$ and $\bv^*$. Furthermore, we must have $\bz= \bu^* 10\bv^*=\bu\bv$, which contradicts our assumption that $\bz \neq \bu \bv$. 
	Hence, the result holds for $t=1$.

	For $t \geq 2$, we prove by induction on $n$. The base case is when $n=1$, i.e. $\bx=0$ and $\by=1$, which is when $\bu$ and $\bv$ are empty strings. In this case the statement is trivial.
	
	Suppose that for any pair of binary words $\bx'=\bu'1\bv'$ and $\by'=\bu'0\bv'$ of length $n \leq k-1$, we have ${\cal D}_t(\bx') \cap {\cal D}_t(\by')= {\cal D}_{t-1}(\bu'\bv')$ for any $t \geq 2$. Let $\bx=\bu1\bv$ and $\by=\bu0\bv$ be binary words of length $k$. Let $t \geq 2$ and $\cS={\cal D}_t(\bx) \cap {\cal D}_t(\by)$. Now, we want to consider several cases for the prefix $\bu$. Suppose $\bu$ is a nonempty binary word, and suppose further that $\bu=a\bu^*$, for $a\in \{0,1\}$. Then consider the following disjoint subsets
	\begin{align*}
	\cS^{\overline{a}} &= {\cal D}_t(\bx)^{\overline{a}} \cap {\cal D}_t(\by)^{\overline{a}}={\cal D}_t(a\bu^*0\bv)^{\overline{a}} \cap {\cal D}_t(a\bu^*1\bv)^{\overline{a}} \\
	&={\cal D}_{t-1}(\bu^*0\bv)^{\overline{a}} \cap {\cal D}_{t-1}(\bu^*1\bv)^{\overline{a}}={\cal D}_{t-2}(\bu^* \bv)^{\overline{a}}\\
	&\subset {\cal D}_{t-1}(\bu\bv),
	\end{align*}
	where the last equality follows from our induction hypothesis if $t\geq 3$ or from our first result if $t=2$,
	\begin{align*}
	\cS^a &= {\cal D}_t(\bx)^a \cap {\cal D}_t(\by)^a={\cal D}_t(a\bu^*0\bv)^a \cap {\cal D}_t(a\bu^*1\bv)^a \\
	&=a \circ \left( {\cal D}_t(\bu^*0\bv) \cap {\cal D}_t(\bu^*1\bv) \right) = a \circ {\cal D}_{t-1}(\bu^* \bv) \\
	&\subset {\cal D}_{t-1}(\bu \bv),
	\end{align*}
	where the last equality holds because of our induction hypothesis. Therefore we have $\cS=\cS^a \cup \cS^{\overline{a}} \subset {\cal D}_{t-1}(\bu \bv) $. Furthermore it is clear that ${\cal D}_{t-1}(\bu \bv) \subset {\cal D}_t(\bu 1\bv) \cap {\cal D}_t(\bu 0 \bv)=\cS$. Hence we have $\cS={\cal D}_{t-1}(\bu \bv)$.
	
	Suppose $\bv$ is a nonempty binary subword, i.e. $\bv=\bv^*a$, for $a\in \{0,1\}$. Then similarly to the above, by considering $\cS_a$ and $\cS_{\overline{a}}$, we can also show that $\cS={\cal D}_{t-1}(\bu \bv)$.
	
	Therefore we are left with the case when $\bu$ and $\bv$ are both empty subwords, which is already shown as the base case.
\end{proof}

Next, we consider the case where the words are Type-B-confusable with the subwords $\bu$ and $\bw$ being empty. 
%In our proof, we appeal to the following results on $D_t(n)$ and $\nu_t(n)$.
%
%\begin{lemma}[\hspace{-0.3pt}\cite{Levenshtein.2001,Liron.2015}]\label{lemma:dtn}
%	 Let $1\le t\le n$.
%	\begin{enumerate}[(i)]
%	\item In addition to \eqref{eq:read-wholespace}, 
%	we have that 
%	\begin{align*}
%	D_t(n) &= D_{t}(n-1)+D_{t-1}(n-2),\\
%	\nu_t(n) &=\nu_{t}(n-1)+\nu_{t-1}(n-2).
%	\end{align*}
%	\item $D_t(n) \geq D_{t-i}(n-i)$ and 
%	$\nu_t(n) \geq \nu_{t-i}(n-i)$ for $ i \leq t$.
%	\item Let $\bx=\bu aa\bv$ where $\bu,\bv$ are subwords and $a \in \{0,1\}$. If $\by= \bu a\overline{a\bv}$, then $|\D_t(\bx)| \leq |\D_t(\by)|$.
%	\end{enumerate}
%\end{lemma}
%
%We also require the following lemma.	
%\begin{lemma}\label{lemma:startswithaletter}
%	Let $a\in\{0,1\}$ and $\bv$ be a word. Then we have that
%	\begin{equation}\label{eq:n-1}
%	|D_t(\bv)^a|\le D_t(|\bv|-1).
%	\end{equation}
%\end{lemma}
%
%\begin{proof}
%{\color{blue} If $a$ does not appear in $\bv$, then the inequality is trivial. Now suppose that $\bv={\overline{a}^m}a\bv^*,$ for some $m \geq 0$, and subword $\bv^*$. Then we have $|D_t(\bv)^a|=|D_t({\overline{a}^m}a\bv^*)^a|=|D_{t-m}({\bv^*})|\leq D_{t-m}(|\bv^*|)=D_{t-m}(|\bv|-m-1) \leq D_{t}(|\bv|-1)$, where the last inequality follows from Lemma \ref{lemma:dtn}.
%
%}
%\end{proof}

\begin{lemma}\label{lemma:firsttwoforms}
	Let $\bx$ and $\by$ be binary words of the form
	\[
	\bx= a \overline{a} \bv b  \quad \text{and}\quad \by= \overline{a} \bv b\overline{b},
	\] or vice versa,
	for some subword $\bv$ of length $n-3$ and $a,b \in \{0,1\}$.
	If $\D_1(\bx)\cap\D_1(\by)=\{\bz\}$, 
	then $|D_t(\bx) \cap D_t(\by)| \leq |D_{t-1}(\bz)| +\nu_{t-1}(n-3)$.
\end{lemma}

\begin{proof}
Let $\cS=D_t(\bx) \cap D_t(\by)$. We split this into two cases. \begin{enumerate}[(i)]
	\item If $b=a$, and hence $\bx=a\overline{a}\bv a$ and $\by=\overline{a}\bv a\overline{a}$.\\
	Note that $\bz= \overline{a}\bv a$. We consider the following three subsets:
	\begin{align*}
	\cS^{\overline{a}} &\subset {\cal D}_t(\bx)^{\overline{a}} \subset {\cal D}_{t-1}(\overline{a}\bv a) = {\cal D}_{t-1}(\bz) ,\\
	\cS_a &\subset {\cal D}_t(\by)_a \subset {\cal D}_{t-1}(\overline{a}\bv a) = {\cal D}_{t-1}(\bz),\\
	\cS_{\overline{a}}^a&={\cal D}_t(a\overline{a}\bv a)_{\overline{a}}^a \cap {\cal D}_t(\overline{a}\bv a\overline{a})_{\overline{a}}^a\\
	&=a \circ {\cal D}_{t-1}(\overline{a}\bv)_{\overline{a}} \cap {\cal D}_{t-1}(\bv a)^a \circ \overline{a}.
	\end{align*}
	Observe that if $\bv=\bv^*a$, for some subword $\bv^*$ then 
	\begin{align*}
	|\cS_{\overline{a}}^a| &\leq |{\cal D}_{t-1}(\overline{a}\bv)_{\overline{a}}|=|{\cal D}_{t-1}(\overline{a}\bv^*a)_{\overline{a}}|\\
	&=|{\cal D}_{t-2}(\overline{a}\bv^*)_{\overline{a}}|\leq D_{t-2}(|\overline{a}\bv^*|-1) \leq D_{t-2}(n-4).
	\end{align*}
	Similarly if $\bv=\overline{a}\bv^*$ for some subword $\bv^*$ then
	\begin{align*}
	|\cS_{\overline{a}}^a| &\leq |{\cal D}_{t-1}(\bv a)^a|=|{\cal D}_{t-1}(\overline{a}\bv^*a)^a|\\
	&=|{\cal D}_{t-2}(\bv^*a)^a|\leq D_{t-2}(|\bv^*a|-1)\leq D_{t-2}(n-4).
	\end{align*}
	In both cases, it follows from Lemma \ref{lemma:dtn} that $|\cS^a_{\overline{a}}|\leq D_{t-2}(n-4)=D_{t-2}(n-5)+D_{t-3}(n-6)\leq 2D_{t-2}(n-5)=\nu_{t-1}(n-3)$.  Hence $|\cS|=|\cS^{\overline{a}} \cup \cS_a|+|\cS_{\overline{a}}^a|\leq |{\cal D}_{t-1}(\bz)|+\nu_{t-1}(n-3) $.
	
	It remains to show for the case when $\bv$ is empty or $\bv=a\bv^*\overline{a}$. The former case would imply that $\bx$ and $\by$ are Type-A-confusable, which contradicts our assumption. While in the latter case, we have $|\cS_{\overline{a}}^a|=|{\cal D}_{t-1}(\overline{a}a\bv^*)\cap {\cal D}_{t-1}(\bv^*\overline{a}a)|$. It can be shown that $\overline{a}a\bv^*$ is equal to $\bv^*\overline{a}a$ if and only if $\bv^*=(\overline{a}a)^m$ for $m \geq 0$, in which case $\bx=a\overline{a}a(\overline{a}a)^m\overline{a}a$ and $\by=\overline{a}a(\overline{a}a)^m\overline{a}a\overline{a}$ would be Type-A-confusable and by Lemma \ref{char:type-A} contradicts our assumption. Therefore we know that $\overline{a}a\bv^*$ and $\bv^*\overline{a}a$ are distinct binary words, and thus $|{\cal D}_{t-1}(\overline{a}a\bv^*)\cap {\cal D}_{t-1}(\bv^*\overline{a}a)|\leq \nu_{t-1}(n-3)$. Therefore in this case also, $|\cS|=|\cS^{\overline{a}} \cup \cS_a|+|\cS_{\overline{a}}^a|\leq |{\cal D}_{t-1}(\bz)|+\nu_{t-1}(n-3)$.\\
	
	\item If $b=\overline{a}$, and hence $\bx=a\overline{a}\bv \overline{a}$ and $\by=\overline{a}\bv \overline{a}a$.\\
	Note that $\bz= {\overline{a}}\bv{\overline{a}}$. We consider the following three subsets:
	\begin{align*}
	\cS^{\overline{a}} &\subset {\cal D}_t(\bx)^{\overline{a}} \subset {\cal D}_{t-1}({\overline{a}}\bv{\overline{a}}) = {\cal D}_{t-1}(\bz) ,\\
	\cS_{\overline{a}} &\subset {\cal D}_t(\by)_{\overline{a}} \subset {\cal D}_{t-1}({\overline{a}}\bv{\overline{a}}) = {\cal D}_{t-1}(\bz),\\
	\cS_a^a&={\cal D}_t(a{\overline{a}}\bv {\overline{a}})_a^a \cap {\cal D}_t({\overline{a}}\bv{\overline{a}}a)_a^a\\
	&=a \circ {\cal D}_{t-1}({\overline{a}}\bv)_a \cap {\cal D}_{t-1}(\bv{\overline{a}})^a \circ a.
	\end{align*}
	Observe that if $\bv=\bv^*{\overline{a}}$, for some subword $\bv^*$ then 
	\begin{align*}
	|\cS_a^a|&\leq|{\cal D}_{t-1}({\overline{a}}\bv)_a|=|{\cal D}_{t-1}({\overline{a}}\bv^*{\overline{a}})_a|\\
	&=|{\cal D}_{t-2}({\overline{a}}\bv^*)_a|\leq D_{t-2}(|\overline{a}\bv^*|-1)=D_{t-2}(n-4).
	\end{align*}Similarly if $\bv={\overline{a}}\bv^*$ for some subword $\bv^*$ then  \begin{align*}
	|\cS_a^a|&\leq|{\cal D}_{t-1}(\bv{\overline{a}})^a|=|{\cal D}_{t-1}({\overline{a}}\bv^*{\overline{a}})^a|\\
	&=|{\cal D}_{t-2}(\bv^*{\overline{a}})^a|\leq D_{t-2}(|\bv^*{\overline{a}}|-1)=D_{t-2}(n-4).
	\end{align*}In both cases, exactly as the previous case it follows from Lemma \ref{lemma:dtn} that $|\cS^a_{\overline{a}}|\leq D_{t-2}(n-4)\leq \nu_{t-1}(n-3)$, and hence $|\cS|=|\cS^{\overline{a}} \cup \cS_{\overline{a}}|+|\cS_a^a|\leq |{\cal D}_{t-1}(\bz)|+\nu_{t-1}(n-3) $.
	
	If $\bv$ is empty, then the statement is trivial. If $\bv=a$, then $\bx$ and $\by$ would be Type-A-confusable, which contradicts our assumption. It remains to show for the case $\bv=a\bv^*a$. In which case, $|\cS_a^a|=|{\cal D}_{t-1}({\overline{a}}a\bv^*)\cap {\cal D}_{t-1}(\bv^*a{\overline{a}})|$. It can be shown that ${\overline{a}}a\bv^*$ is equal to $\bv^*a{\overline{a}}$ if and only if $\bv^*={\overline{a}}(a{\overline{a}})^m$ for $m \geq 0$, in which case $\bx=a{\overline{a}}a{\overline{a}}(a{\overline{a}})^ma{\overline{a}}$ and $\by={\overline{a}}a{\overline{a}}(a{\overline{a}})^ma{\overline{a}}a$ would be Type-A-confusable, and by Lemma \ref{char:type-A} contradicts our assumption. Therefore, we know that ${\overline{a}}a\bv^*$ and $\bv^*a{\overline{a}}$ are distinct binary words, and thus $|{\cal D}_{t-1}({\overline{a}}a\bv^*)\cap {\cal D}_{t-1}(\bv^*a{\overline{a}})|\leq \nu_{t-1}(n-3)$. Hence, $|\cS|=|\cS^{\overline{a}} \cup \cS_{\overline{a}}|+|\cS_{\overline{a}}^a|\leq |{\cal D}_{t-1}(\bz)|+\nu_{t-1}(n-3)$. \qedhere
\end{enumerate}
\end{proof}

\subsection{Proof of Theorem~\ref{thm:2-reconstruction}}\label{sec:completeproof}

We first consider the case $t=2$ and prove a stronger version of Theorem~\ref{thm:2-reconstruction}.
%Next, we prove a stronger version of Theorem~\ref{thm:2-reconstruction} for the case $t=2$.

\begin{theorem}\label{thm:coverage-t2}
	Let $\bx$ and $\by$ be words of length $n\ge 4$ that are Type-B-confusable.
	If $\D_1(\bx)\cap\D_1(\by)=\{\bz\}$ and then we have that 
	\begin{equation}\label{eq:coverage-t2}
	\D_2(\bx)\cap\D_2(\by) = \D_1(\bz)\cup \T(\bx,\by),
	\end{equation}
	where $|\T(\bx,\by)|\le 2$.
\end{theorem}

\begin{proof}
%	{\color{blue} %Note that by Lemma \ref{lemma:lev1form}, we know that $\bx$ and $\by$ can be either of Hamming distance one, or Type-B-confusable. If $\bx$ and $\by$ have hamming distance one, we know from Lemma \ref{lemma:hamming1forgeneral}, that $D_2(\bx) \cap D_2(\by)= D_{1}(\bz)$, and hence $ \T(\bx,\by)=\emptyset$.
	%Otherwise, if $\bx$ and $\by$ are Type-B-confusable, then without loss of generality, 
	Suppose $\bx=\bu a \overline{a} \bv b \bw$ and $\by=\bu \overline{a} \bv b\overline{b}\bw$, for some subwords $\bu, \bv$ and $\bw$, where $a,b \in \{0,1\}$. Let $\cS=D_2(\bx) \cap D_2(\by)$. Note that $\bz=\bu \overline{a} \bv b w$.
	
	We are going to show the result by induction on $n$. The base case is when $\bu$ and $\bw$ are empty subwords, which from the proof of Lemma \ref{lemma:firsttwoforms}, we can obtain that $D_2(\bx) \cap D_2(\by) =D_{1}(\bz) \cup \T(\bx,\by),$ where $|\T(\bx,\by)|\leq \nu_1(n-3)=2$, where the last equality comes from  Lemma \ref{lemma:dtn}. Suppose the statement is true for length $n \leq k-1$, we want to show for $n=k$. Now, we want to consider several cases for the prefix $\bu$. Suppose $\bu$ is a nonempty prefix, i.e. $\bu=c\bu^*$, for some subword $\bu^*$ and $c\in \{0,1\}$.
	
	Consider the following
	\begin{align}
	\cS^{\overline{c}} &=\D_2(c\bu^*a{\overline{a}}\bv b\bw)^{\overline{c}} \cap \D_2(c\bu^*{\overline{a}}\bv b{\overline{b}}\bw)^{\overline{c}} \nonumber\\
	&=\D_{1}(\bu^*a{\overline{a}}\bv b\bw)^{\overline{c}} \cap \D_{1}(\bu^*{\overline{a}}\bv b{\overline{b}}\bw)^{\overline{c}}\nonumber \\
	&=\D_0(\bu^*{\overline{a}}\bv b\bw)^{\overline{c}}=\D_1(c\bu^* \overline{a} \bv b w)^{\overline{c}}=\D_1(\bz)^{\overline{c}},\label{eq:part1}
	\end{align}
	where the third equality holds because $|\D_1(\bx)\cap\D_1(\by)|=1$, and
	\begin{align}
	\cS^c &=\D_2(c\bu^*a{\overline{a}}\bv b\bw)^c \cap \D_2(c\bu^*{\overline{a}}\bv b{\overline{b}}\bw)^c \nonumber\\
	&=c \circ \left(\D_2(\bu^*a{\overline{a}}\bv b\bw) \cap \D_2(\bu^*{\overline{a}}\bv b{\overline{b}}\bw)\right)\nonumber \\
	&= c \circ \left( \D_1(\bu^* \overline{a} \bv b \bw)\cup \T(\bu^*a{\overline{a}}\bv b\bw,\bu^*{\overline{a}}\bv b{\overline{b}}\bw)
	\right)\nonumber\\
	&=\D_1(\bz)^c \cup \T(\bx,\by)\label{eq:part2}
	\end{align}
	where the third equality follows from our induction hypothesis, and $\T(\bx,\by)=c \circ \T(\bu^*a{\overline{a}}\bv b\bw,\bu^*{\overline{a}}\bv b{\overline{b}}\bw)$, and hence $|\T(\bx,\by)| \leq 2$, from our hypothesis.
	Combining equations \eqref{eq:part1} and \eqref{eq:part2}, we have $D_2(\bx) \cap D_2(\by)=\cS=\cS^c \cup \cS^{\overline{c}} =\D_1(\bz) \cup \T(\bx,\by),$ where $|\T(\bx,\by)|\leq 2$. 
	Hence, induction on $n$ is complete.
	%Hence we have shown the theorem by induction on $n$.
%}
\end{proof}

Next, we make the following observation on the word $\bz$ that lies in the intersection of the single-deletion balls.

\begin{lemma}\label{lem:z-not-alternating}
	If $\bx$ and $\by$ are Type-B-confusable and $\D_1(\bx)\cap \D_1(\by)=\{\bz\}$, 
	then $\bz$ is not alternating.
\end{lemma}

\begin{proof}
	 Suppose $\bx=\bu a \overline{a} \bv b \bw$ and $\by=\bu \overline{a} \bv b\overline{b}\bw$, for some subwords $\bu, \bv$ and $\bw$, where $a,b \in \{0,1\}$. Suppose that $\bz=\bu \overline{a} \bv b w$ is an alternating sequence. This means $\overline{a} \bv b$ is an alternating subword, and hence $a\overline{a} \bv b$ (and $\overline{a} \bv b \overline{b}$), which is a subwords of $\bx$ (and $\by$, respectively) is an alternating sequence as well. This implies that $\bx=\bu \ba \bw$ and $\by=\bu \overline{\ba} \bw$ are Type-A-confusable, where $\ba=a\overline{a} \bv b$, and hence $|\D_2(\bx) \cap \D_2(\by)| =2$ by Lemma \ref{char:type-A}, which contradicts our assumption.
\end{proof}

Finally, we prove the main result of this section.
% Theorem~\ref{thm:2-reconstruction}.
%the main result \eqref{eq:coverage2}.

\begin{proof}[Proof of Theorem~\ref{thm:2-reconstruction}]
	From Lemma \ref{lemma:lev1form}, we know that if $\D_1(\bx)\cap\D_1(\by)=\{\bz\}$, then there are two possibilities. First possibility is when $\bx$ and $\by$ have Hamming distance one, which by Lemma \ref{lemma:hamming1forgeneral}, implies that $|\D_t(\bx)\cap\D_t(\by)|= |\D_{t-1}(\bz)| < |\D_{t-1}(\bz)|+\nu_{t-1}(n-3)\leq D_{t-1}(n-1)+\nu_{t-1}(n-3)$, since $\nu_{t-1}(n-3)>0$ for $t \geq 2$ and $n\geq 6$.
	
	Second possibility is when $\bx$ and $\by$ are Type-B-confusable. Note that for $t=2$, from Theorem~\ref{thm:coverage-t2}, we know that $|\D_2(\bx) \cap \D_2(\by)| \leq |\D_1(\bz)|+2 < D_1(n-1)+\nu_1(n-3)$, where the strict inequality comes from 
	%Lemma \ref{lem:z-not-alternating} and Theorem \ref{thm:strictlyless}. 
	Lemmas~\ref{lem:z-not-alternating} and~\ref{lemma:dtn}(v).
	Now, we only need to show for $t \geq 3$.
	
	Without loss of generality, let $\bx=\bu a {\overline{a}} \bv b \bw \quad \text{and} \quad \by=\bu {\overline{a}} \bv b{\overline{b}}\bw$. Let $\cS=\D_t(\bx) \cap \D_t(\by)$. Note that $\bz=\bu {\overline{a}} \bv b \bw$.
	
	We are going to show the result by induction on $n$. The base case is when $\bu$ and $\bw$ are empty. In this case, from Lemma \ref{lemma:firsttwoforms}, we have $|\cS|  \leq |D_{t-1}(\bz)| +\nu_{t-1}(n-3)$, 
	%and further from Lemma \ref{lem:z-not-alternating} and Theorem \ref{thm:strictlyless}, 
	and further from Lemmas~\ref{lem:z-not-alternating} and~\ref{lemma:dtn}(v), we have the desired result. 
	Suppose the statement is true for length $n \leq k-1$, we want to show for $n=k$. Now, we want to consider several cases for the prefix $\bu$. Suppose $\bu$ is a nonempty subword, i.e. $\bu=c\bu^*$, for some subword $\bu^*$ and $c\in \{0,1\}$.
	
	Consider the following,
	\begin{align}
	|\cS^{\overline{c}}| &= |\D_t(c\bu^*a{\overline{a}}\bv b\bw)^{\overline{c}}\cap \D_t(c\bu^*{\overline{a}}\bv b{\overline{b}}\bw)^{\overline{c}}| \nonumber  \\
	&=|\D_{t-1}(\bu^*a{\overline{a}}\bv b\bw)^{\overline{c}} \cap \D_{t-1}(\bu^*{\overline{a}}\bv b{\overline{b}}\bw)^{\overline{c}}|, \label{eq:firstsubset}\\
	|\cS^c| &= |\D_t(c\bu^*a{\overline{a}}\bv b\bw)^c \cap \D_t(c\bu^*{\overline{a}}\bv b{\overline{b}}\bw)^c| \nonumber \\
	&=|\D_t(\bu^*a{\overline{a}}\bv b\bw) \cap \D_t(\bu^*{\overline{a}}\bv b{\overline{b}}\bw)|\nonumber \\
	&\leq D_{t-1}(n-2)+\nu_{t-1}(n-4) 
	\label{eq:secondsubset},
	\end{align} where the last inequality holds from our induction hypothesis.
	Now, consider the following cases
	
	{\bf Case 1}: If ${\overline{c}}$ does not appear in $\bu^*$ and ${\overline{c}}={\overline{a}}$.\\From \eqref{eq:firstsubset}, we have $|\cS^{\overline{c}}|\leq |\D_{t-1}(\bu^*a{\overline{a}}\bv b \bw)^{\overline{c}}|\leq |\D_{t-2-|\bu^*|}(\bv b\bw)| \leq D_{t-2-|\bu^*|}(n-|\bu^*|-3)\leq D_{t-2}(n-3)$, where the last inequality follows from Lemma \ref{lemma:dtn}. For $t \geq 3,n \geq 6$, we have $\nu_{t-2}(n-5)>0$. Thus $|\cS^{\overline{c}}| < D_{t-2}(n-3)+\nu_{t-2}(n-5)$, and combined with \eqref{eq:secondsubset}, we have $|\cS|=|\cS^c|+|\cS^{\overline{c}}| < D_{t-2}(n-3)+\nu_{t-2}(n-5) +D_{t-1}(n-2)+\nu_{t-1}(n-4)=D_{t-1}(n-1)+\nu_{t-1}(n-3)$, where the last equality follows from Lemma \ref{lemma:dtn}
	
	{\bf Case 2}: If ${\overline{c}}$ does not appear in $\bu^*$ and ${\overline{c}}=a$.\\From \eqref{eq:firstsubset}, we have $|\cS^{\overline{c}}| \leq |\D_{t-1}(\bu^*{\overline{a}}\bv b{\overline{b}} \bw)^{\overline{c}}|\leq |\D_{t-2-|\bu^*|}(\bv b{\overline{b}}\bw)^{a} |\leq D_{t-2-|\bu^*|}(|\bv b{\overline{b}}\bw|-1)\leq D_{t-2-|\bu^*|}(n-|\bu^*|-3)\leq D_{t-2}(n-3)$, 
	%where the third inequality comes from Lemma \ref{lemma:startswithaletter} and the last inequality follows from Lemma \ref{lemma:dtn}. 
	where the third and last inequalities come from Lemma~\ref{lemma:dtn}.
	Thus similar to Case 1, we have $|\cS^{\overline{c}}| < D_{t-2}(n-3)+\nu_{t-2}(n-5)$, and therefore $|\cS|<D_{t-1}(n-1)+\nu_{t-1}(n-3)$.
	
	{\bf Case 3}: : If $t < \frac{n}{2}$ and ${\overline{c}}$ appears in $\bu^*$ i.e. $\bu^*=c^m \overline{c} \bu'$, for some binary sequence $\bu'$ and $m\geq 0$.\\Note that $t-m-1 < \frac{n-m-2}{2}$, therefore from \eqref{eq:firstsubset}, we have $|\cS^{\overline{c}}|=|\D_{t-m-1}(\bu'a{\overline{a}}\bv b \bw) \cap \D_{t-m-1}(\bu'{\overline{a}}\bv b{\overline{b}} \bw)| < D_{t-m-2}(n-m-3) +\nu_{t-m-2}(n-m-5) \leq D_{t-2}(n-3)+\nu_{t-2}(n-5)$,
	where the first inequality holds because $\bu'a{\overline{a}}\bv b\bw$ and $\bu'{\overline{a}}\bv b{\overline{b}}\bw$ are \textit{Type-B confusable} and hence we can use our induction hypothesis, and the last inequality follows from Lemma \ref{lemma:dtn}. Combined with \eqref{eq:secondsubset}, we have $|\cS|=|\cS^c|+|\cS^{\overline{c}}| < D_{t-2}(n-3)+\nu_{t-2}(n-5) +D_{t-1}(n-2)+\nu_{t-1}(n-4)=D_{t-1}(n-1)+\nu_{t-1}(n-3)$, where the last equality follows from Lemma \ref{lemma:dtn}.
	
	{\bf Case 4}: : If $t \geq \frac{n}{2}$ and ${\overline{c}}$ appears in $\bu^*$ i.e. $\bu^*=c^m \overline{c} \bu'$, for some binary sequence $\bu'$ and $m\geq 0$. Similar to Case 3, using induction hypothesis , we also have $|\cS^{\overline{c}}|=|\D_{t-m-1}(\bu'a{\overline{a}}\bv b \bw) \cap \D_{t-m-1}(\bu'{\overline{a}}\bv b{\overline{b}} \bw)| \leq D_{t-m-2}(n-m-3) +\nu_{t-m-2}(n-m-5) \leq D_{t-2}(n-3)+\nu_{t-2}(n-5)$, and therefore $|\cS| \leq D_{t-1}(n-1)+\nu_{t-1}(n-3)$.
	
	In all cases, we have shown that the statement is true. Now suppose $\bw$ is a nonempty binary sequence, i.e. $\bw=\bw^*c$, for some binary sequence $\bw^*$ and $c\in \{0,1\}$, then similarly to the above, by considering $|\cS_c|$ and $|\cS_{\overline{c}}|$, we can also show that the statement is true.
	
	Therefore we are left with the case when $\bu$ and $\bw$ are both empty strings, which is already covered in the base case.
\end{proof}

\subsection{Improvements when $t=2$}
\label{sec:t2}

To conclude this section, we focus on the case $t=2$ and show that by controlling the parameter $P$ in Definition~\ref{def:csvt}, we are able to bound the number of noisy reads required to reconstruct a codeword. To do so, we make the following simple observation.

\begin{lemma}\label{lem:runs-csvt}
	Let $\bz$ be a word of length $n$.
	If the length of any alternating run in a word $\bz$ is at most $P$, then the number of runs in $\bz$ is at most $n-\ceil*{n/P}+1$.
	Therefore, $|\D_t(\bz)|\le n-\ceil*{n/P}+1$.
\end{lemma}
\begin{proof}
	Let $S=\{i \in \mathbb{Z} : \bz_i=\bz_{i+1}\}$. We order the elements of $S$ and call them $s_1, s_2, ..., s_{|S|}$ from smallest to biggest. We want to show that $|S| \geq \ceil*{n/P}-1$. Note that $s_{i+1}-s_i \leq P$ for all $i\geq 1$, $s_1\leq P$ and $s_{|S|}\geq n-P$, since otherwise there would be an alternating run of length more than $P$.
	
	Suppose on the contrary that $|S| <\ceil*{n/P}-1$, this implies that $s_{|S|}=s_1+\sum_{i=1}^{|S|-1}{s_{i+1}-s_i} \leq |S|P\leq \left(\ceil*{n/P}-2\right)P <(n/P +1 -2)P=n-P $, which contradicts that $s_{|S|}\geq n-P$. Therefore $|S| \geq \ceil*{n/P}-1$, and hence the number of runs in $\bz$ is at most $n-\ceil*{n/P}+1$.
\end{proof}

Recall that by design, the length of any alternating run of any codeword $\bx$ in a constrained SVT code is at most $P$. 
Hence, the same property holds for any word $\bz$ in the single-deletion ball of $\bx$.
So, we can apply Lemma~\ref{lem:runs-csvt} and provide a tighter bound on the size of $\D_1(\bz)$

\begin{proposition}
	For any $c\in\bbZ_{1+P/2}$ and $d\in\bbZ_2$, the constrained SVT code $\CSVT(n,P;c,d)$ is an $(n,N_P; \D_2)$-reconstruction code 
	where $N_P=\max\{n-\ceil*{(n-1)/P}+3, 7\}$.
\end{proposition}

\begin{proof}
	Let $\bx$ and $\by$ be distinct codewords 
	Then $|\D_1(\bx)\cap \D_1(\by)|\le 1$ and it remains to show that $|\D_t(\bx)\cap \D_t(\by)|< N_P$ .
	
	When the intersection is empty, Theorem~\ref{thm:1-reconstruction} states that  $|\D_t(\bx)\cap \D_t(\by)|\le 6<N_P$.
	
	When $|\D_1(\bx)\cap \D_1(\by)|=1$, let $\bz$ be the word.
	Then since $\bz$ is a subword of $\bx$, 
	the alternating run of $\bz$ is of length at most $P$ and the number of runs of $\bz$ is at most $n-1-\ceil*{(n-1)/P}$.
	Applying \eqref{eq:coverage-t2}, we have that $|\D_t(\bx)\cap \D_t(\by)|<N_P$, as required.
\end{proof}

Let $P\ge 4$. It is well-known (see for example, \cite{chee2018coding}) that the number of length-$n$ words whose 2-periodic run is at most $P$ is $4F_{P-1}(n-2)$, where
\[
F_\ell(n) = 
\begin{cases}
2^n, & \text{if } 0\le n\le \ell-1,\\
\sum_{i=1}^\ell F_\ell(n-i), &\text{otherwise.}
\end{cases}
\]

Hence, we have the following lower bound on the size of a reconstruction code.

\begin{corollary}\label{cor:csvt}
	For $P\ge 4$, set $N_P=\max\{n-\ceil*{(n-1)/P}+3, 7\}$. Then there exists an $(n,N_P; \D_2)$-reconstruction code of size at least $4 F_{P-1}(n-2)/(P+2)$.
\end{corollary}

To end this section, for codelengths $n\in\{127,255,1023\}$, we vary the parameter $P$ in the constrained SVT codes and compute the corresponding values of $N_P$ and redundancy. The numerical results are given in Table~\ref{table:csvt}. 
As expected, as we decrease the value of $P$, the number of required reads also decreases. 
However, the number of redundant bits also increases significantly and 
in this case (where $P$ is small), the VT code uses significantly less redundant bits. 
For completeness, we list the values of read-coverage and redundancy of a VT-code of length $n$ and the space $\{0,1\}^n$ (corresponding to the uncoded case).

\begin{table}[!h]
\begin{center}
	\renewcommand{\arraystretch}{1.1}
	
	\begin{tabular}{|c|r|R{1.5cm}|r|l|}
		\hline
		$n$ & $P$ & $N_P$ / Read Coverage & Redundancy & Remarks \\
		\hline\hline
		127 & -- & 7 & 7.00 & VT code \\
		127 & 6 & 109 & 6.016 & -- \\
		127 & 8 & 114 & 4.018 & -- \\
		127 & 10 & 117 & 3.753 & -- \\
		127 & -- & 250 & 0.00 & $\{0,1\}^n$\\
		\hline
		\hline
		255 & -- & 7 & 8.00 & VT code \\
		255 & 8 & 226 & 4.762 & -- \\
		255 & 10 & 232 & 3.935 & -- \\
		255 & 12 & 236 & 3.894 & -- \\
		255 & -- & 506 & 0.00 & $\{0,1\}^n$\\
		\hline
		\hline
		1023 & -- & 7 & 10.00 & VT code \\
		1023 & 8 & 898 & 9.22 & -- \\
		1023 & 10 & 923 & 5.03 & -- \\
		1023 & 12 & 940 & 4.17 & -- \\
		1023 & 14 & 953 & 4.09 & -- \\
		1023 & -- & 2042 & 0.00 & $\{0,1\}^n$\\
		\hline		
	\end{tabular}
\end{center}
\caption{List of constrained SVT codes and their read coverage and redundancy}\label{table:csvt}
\end{table}

%\section{Conclusion}

%\section{TODO}

\newpage

\end{document}